\documentclass{article16}
\usepackage[]{algorithm2e}
\usepackage{filecontents}
\usepackage{mathtools}

\DeclarePairedDelimiter\abs{\lvert}{\rvert}%
\DeclarePairedDelimiter\norm{\lVert}{\rVert}%

\makeatletter
\let\oldabs\abs
\def\abs{\@ifstar{\oldabs}{\oldabs*}}
\let\oldnorm\norm
\def\norm{\@ifstar{\oldnorm}{\oldnorm*}}
\makeatother

\title{Minimum Hidden Guarding of Histogram Polygons}
\author{Hamid Hoorfar \thanks{Department of Computer Engineering and Information Technology, Amirkabir University of Technology (Tehran Polytechnic), {\tt \{hoorfar,ar.bagheri\}@aut.ac.ir}}\and Alireza Bagheri\footnotemark[1]~\footnote{\textit{Coresponding author}.}
 }

\begin{document}
\maketitle

\begin{abstract}
A hidden guard set $ G $ is a set of point guards in polygon $ P $ that all points of the polygon are visible from some guards in $ G $ under the constraint that no two guards may see each other. In this paper, we consider the problem for finding minimum hidden guard sets in histogram polygons under orthogonal visibility. Two points $ p $ and $ q $ are orthogonally visible if the orthogonal bounding rectangle for $ p $ and $ q $ lies within $ P $. It is known that the problem is NP-hard for simple polygon with general visibility and it is true for simple orthogonal polygon. We proposed a linear time exact algorithm for finding minimum hidden guard set in histogram polygons under orthogonal visibility. In our algorithm, it is allowed that guards place everywhere in the polygon. 

\end{abstract}

\section{Introduction}

In the standard version of the art gallery problem one is given a simple polygon $ P $ in the plane that needs to be guarded by
a set of point guards~\cite{urrutia2000art}. In other words, we want to find a set of point guards such that every point in $ P $ is seen by at least one of the guards, where a guard $ g $ sees a point $ p $ if the segment $ gp $ is contained in $ P $. Lee and Lin~\cite{lee1986computational}
proved that finding the minimum number of point guards needed to guard an arbitrary polygon is NP-hard. The art gallery problem is also NP-hard for orthogonal polygons~\cite{schuchardt1995two} 
and it even remains NP-hard for monotone polygons~\cite{schuchardt1995two}. In some papers, it is assumed that the visibility is in orthogonal model(r-visibility) instead of standard line visibility. In standard, points $ p $ and $ q $ are visible from each other in the polygon $ P $, if line segment $ pq $ is contained in $ P $. So, in orthogonal visibility, they are r-visible (orthogonally visible) to each other if the axis-parallel rectangle spanned by the points is contained in $ P $~\cite{o2004visibility}. Worman and Keil~\cite{worman2007polygon} studied the orthogonal polygon decomposition problem that is equivalent to the orthogonal art gallery problem and showed that the problem is polynomially solvable for r-visibility. The time complexity of their algorithm is $ O(n^{17}poly \log n) $. Gewali and et. al. ~\cite{gewali1996placing} presented an $ O(n) $ time algorithm for covering a monotone orthogonal polygon with the minimum number of orthogonal star-shaped polygons. Finding minimum covers by star-shaped polygons is equivalent to finding the minimum number of guards needed so that every point in the polygon is visible to at least one guard. Palios and Tzimas~\cite{palios2014minimum} considered the the problem on simple class-3 orthogonal polygons, i.e., orthogonal polygons that have dents along at most $ 3 $ different orientations. They gave an output-sensitive $ O(n+k \log ⁡k) $-time algorithm where $ k $ is the size of a minimum r-star cover. Beidl and Mehrabi~\cite{biedl2016r} showed that problem is NP-hard on orthogonal polygons with holes. A polygon is called \textit{tree polygon} If dual graph of the polygon is a tree. Also, They presented an algorithm for tree polygon in the linear-time. If the guards are only allowed to be on vertices(vertex guard variant), discretization combined with iteration  yields an $ O(n^4) $ solution~\cite{couto2007exact} for general orthogonal polygons. In the polygon $ P $, a \textit{hidden set} is defined as a set of points such that no two points in the set are visible to each other. So, a \textit{hidden guard set} is a hidden set which is also a guard set and the entire polygon is visible from some points in it. Finding the minimum or maximum guard set in a polygon can be considered under standard or orthogonal visibility. In the standard Version, Shermer~\cite{shermer1989hiding} was the first to show that the problem of finding a maximum hidden and a minimum hidden guard sets is NP-hard. He established bounds on the maximum size of hidden sets for some polygons. In 1999, Eidenbenz~\cite{eidenbenz1999many} showed that the problem of placing a maximum hiding guards is almost as hard to approximate as the Maximum Clique problem and it cannot be approximated by any polynomial-time algorithm with an approximation ratio of $ n\epsilon $ for some $ \epsilon > 0 $, unless $ P = NP $. He showed that for a simple polygon with holes, it is already true and for a simple polygon without holes there is a constant $ \epsilon > 0 $ such that the problem cannot be approximated with an approximation ratio of $ 1 + \epsilon $. Hurtado and et.al.~\cite{hurtado1996hiding} studied hidden sets of points in arrangements of segments, they  provided bounds for its maximum size. Biswas and et. al.~\cite{biswas1994algorithms} studied similar problem called as the maximum independent set of visibility graph instead of maximum hidden guard sets. They proposed an $ O(n^{3}) $ time algorithm for finding the maximum independent set of the convex visibility graph for a restricted class of simple polygons. A stair-case path is monotone along both the $ x $ axis and $ y $ axis directions. Two points inside an orthogonal polygon are visible under  \textit{stair-case visibility}  if they can be connected by a stair-case path without intersecting its exterior. Eidenbenz and Stamm~\cite{eidenbenz2000maximum} studied the problem of finding a maximum clique in the visibility graph of a simple polygon with $ n $ vertices. They showed that if the input polygons are allowed to contain holes no polynomial time algorithm can achieve an approximation ratio of  $ \frac{n^{1/8-\epsilon}}{4} $ for any $ \epsilon> 0 $, unless $ NP = P $. They proposed an $ O(n^{3}) $ algorithm for the maximum clique problem on visibility graphs for polygons without holes. Their algorithm also finds the maximum weight clique, if the polygon vertices are weighted. They also showed that the problem of partitioning the vertices of a visibility graph of a polygon into a minimum number of cliques is APX-hard for polygons without holes. Later, they are proved tight approximability results on finding minimum hidden guard sets in polygons~\cite{eidenbenz2006finding,eidenbenz2002optimum}. Because of the maximum hidden vertex set problem on a given simple polygon is NP-hard Bajuelos and et. al.~\cite{bajuelos2008estimating} proposes some approximation algorithms to solve this problem, first two based on greedy constructive search and the other are based on the general meta-heuristics Simulated Annealing and Genetic Algorithms. Their solutions are very satisfactory in the sense that they are always close to optimal with an approximation ratio of $ 1.7 $, for arbitrary polygons and with an approximation ratio of $ 1.5 $ for orthogonal polygons. Also, they showed that on average the maximum number of hidden vertices in a simple polygon (arbitrary or orthogonal) with $ n $ vertices is $ n/4 $. Ghosh and et. al.~\cite{ghosh2007visibility} presented an algorithm for computing the maximum clique in the visibility graph $ G $ of a simple polygon $ P $ in $ O(n^{2}e) $ time, where $ n $ and $ e $ are number of vertices and edges of $ G $ respectively and also presented an $ O(ne) $ time algorithm for computing the maximum hidden vertex set in the visibility graph $ G $ of a convex fan $ P $. Kranakis and et. al.~\cite{kranakis2009inapproximability} introduced a new hiding problem as called the max hidden edge set problem. given a polygon $ P $, find a minimum set of edges $ S $ of the polygon such that any straight line segment crossing the polygon intersects at least one of the edges in $ S $. They proved the APX-hardness of the problem for polygons without holes. Cannon and et. al.~\cite{cannon2012hidden} considered guarding classes of simple polygons using mobile guards (polygon edges and diagonals) under the constraint that no two guards may see each other. They provided a nearly complete set of answers to existence questions of open and closed edge, diagonal, and mobile guards in simple, orthogonal, monotone, and star-shaped polygons like that every monotone or star-shaped polygon can be guarded using hidden open mobile (edge or diagonal) guards, but not necessarily with hidden open edge or hidden open diagonal guards. Bajuelos and et. al.~\cite{bajuelos2008escondiendo} studied the maximum hidden vertex set problem for two types of polygons spirals and histograms and proposed a linear algorithm. In this paper, we introduce the problem of finding the minimum number of hidden guards to cover a polygon $ P $ under orthogonal visibility. First, we present a linear-time algorithm for guarding monotone orthogonal polygons, we know that there is a  linear-time algorithm for the problem~\cite{gewali1992covering,lingas2007note}, but, our algorithm is purely geometric. After that, we present an exact algorithm for finding minimum hidden guard set for histogram polygon and a 2-approximation algorithm for monotone orthogonal polygon under orthogonal visibility. In the following, visibility means orthogonal visibility and guarding is under orthogonal visibility and monotonicity means $ x $-monotonicity unless explicitly expressed(mentioned). The time complexity of hidden guarding is remained open, even for monotone polygon.
\begin{figure}
	\centering
	\includegraphics[width=\textwidth]{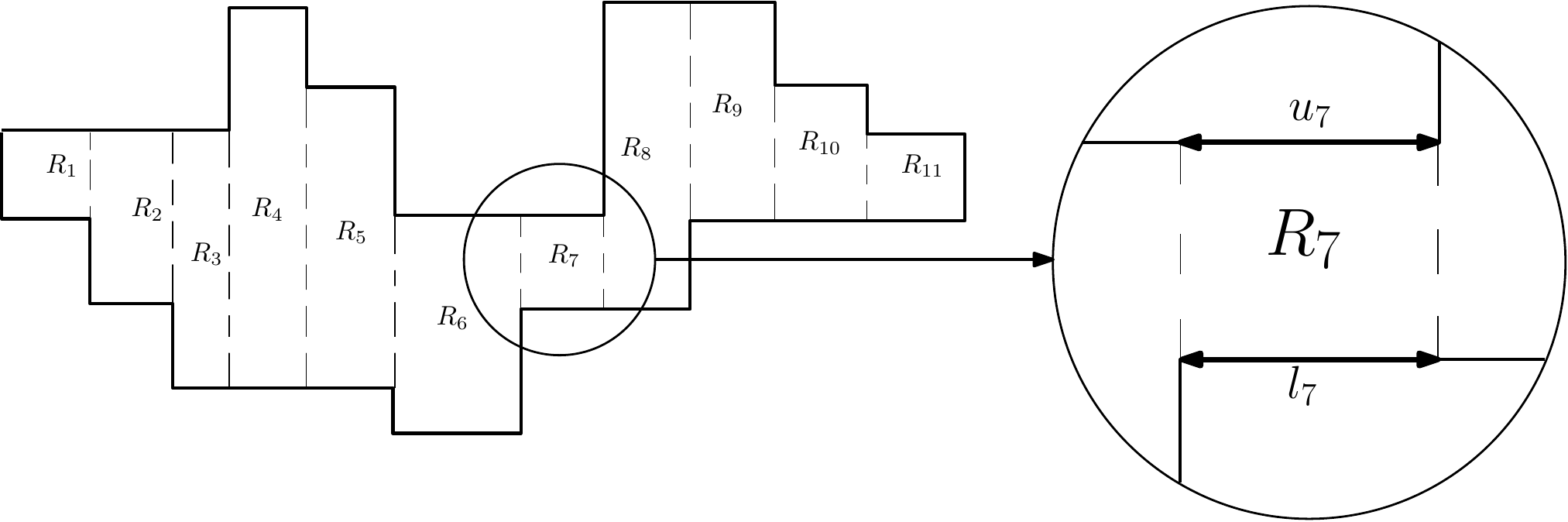}
	\caption{Illumination of the vertical decomposition and its notations.}
	\label{fi:fig1}
\end{figure}
\section{Preliminaries}
\label{ss:ss1}
Assume we decompose a simple orthogonal and $ x $-monotone polygon $ P $ with $ n $ vertices into rectangles obtained by extending the vertical edges incident to the reflex vertices of $ P $. A reflex vertex has an interior angle $ \frac{3\pi}{2} $ while convex vertices have an interior angle of $ \frac{\pi}{2} $. It is clear that every orthogonal polygon with $ n $ vertices has $ \frac{n-4}{2} $ reflex vertices. So, after the decomposition of $ P $, $ \frac{n-2}{2} $ rectangles will be obtained, exactly. Let $R =\{R_{1},R_{2},\dots,R_{m}\} $, where $ m = \frac{n-2}{2} $, be the set of rectangles, ordered from left to right according to $ x $-coordinate of their left edges. In the other words, after the decomposition, we named rectangular parts $ R_{1},R_{2},\dots,R_{m} $ from left to right. It is shown in Figure~\ref{fi:fig1}. We name the upper and lower horizontal edges of $ R_{i} $ by $ u_{i} $ and $ l_{i} $, respectively. Assume that $  U =\{u_{1},u_{2},\dots,u_{m} \}  $ and $ L=\{l_{1},l_{2},\dots,l_{m} \} $, {\small $1\leq i \leq m $}. For a horizontal line segment $ s $, we denote the $ x $-coordinate of the \textbf{left vertex} of $ s $ by $ x(s) $ that means the $ x $-coordinate of line segment $ s $. Also, we denote the $ y $-coordinate of line segment $ s $ by $ y(s) $. For a vertical line segment $ s' $, we denote the $ x $-coordinate of $ s' $ by $ x(s') $. Similarly, for a point $ p_0 $, $ x $-coordinate and $ y $-coordinate of $ p_0 $ is denoted by $ x(p_0) $ and $ y(p_0) $, respectively. Without reducing generality, we suppose that for every two vertical edge $ e_i $ and $ e_j $ ($ i\neq j $), $ x(e_i)\neq x(e_j) $, so, it seems obvious that for all {\small $ 1 \leq i \leq m-1 $}, $ y(u_{i})=y(u_{i+1}) $ or $ y(l_{i})=y(l_{i+1}) $. So, we denote the edge of $ P $ that contains $ u_{i} $  by $ e(u_{i}) $ and the edge of $ P $ that contains $ l_{i} $  by $ e(l_{i}) $. Let $ E_{U}=\{e(u_{i})|1 \leq i \leq m\} $ and $ E_{L}=\{e(l_{i})|1 \leq i \leq m\} $, the sets of edges ordered from left to right. In a set $ E=E_u\cup E_l $ of horizontal edges of $ P $, $ e_{j} $ is named \textit{local maximum} if $ y(e_{j})> y(e_{j-1}) $ and $ y(e_{j})> y(e_{j+1}) $ and similarly, $ e_{k} $ is named \textit{local minimum} if $ y(e_{k})< y(e_{k-1}) $ and $ y(e_{k})< y(e_{k+1}) $. If edge $ \epsilon_1 \in E_U $ be a local maximum, then the internal angles of its both endpoints are equal to $ \frac{\pi}{2} $ and if $ \epsilon_2 $ be a local minimum, then the internal angles of its two endpoints are $ \frac{3\pi}{2} $. If edge $ \epsilon_3 \in E_L $ be a local minimum, then the internal angles of its both endpoints are equal to $ \frac{\pi}{2} $ and if $ \epsilon_4 $ be a local maximum, then the internal angles of its two endpoints are $ \frac{3\pi}{2} $. In this paper, if a horizontal edge $ \epsilon'\in E $ has two endpoints of angle $ \frac{\pi}{2} $, we call it \textit{tooth} and if a horizontal edge $ \epsilon\in E $ has two endpoints of angle $ \frac{3\pi}{2} $, we call it \textit{dent}. $ u_{m}$ (or $ l_{m}$) is named \textit{local maximum} if $ e( u_{m}) $ (or $ e(l_{m}) $) is local maximum, and $ u_{n}$ (or $ l_{n} $) is named \textit{local minimum} if  $ e( u_{n}) $ (or $ e(l_{n} )$) is local minimum. Every rectangle $ R_i $ has height $ h_i=\abs{y(u_i)-y(l_i)} $ , so, in the set $ R $, $ R_{l} $ is named \textit{local maximum} if $ h_l>h_{l-1} $ and $ h_l>h_{l+1} $, and $ R_y $ is named \textit{local minimum} if $ h_y<h_{y-1} $ and $ h_y<h_{y+1} $. The rectangle that has only one neighbor is named \textit{source}, regularly, every rectangle has two neighbor except the first and last ones. So, the first and last rectangle is called source. Two axis-parallel segments $ l $ and $ l' $ are defined as \textit{weak visible} if an axis-parallel line segment could be drawn from some point of $ l $ to some point of $ l' $ that does not intersect $ P $. If polygon $ P $ is $ x $-monotone and also $ y $-monotone, then $ P $ is named \textit{orthoconvex} polygon. 
\begin{figure}
	\centering
	\includegraphics[width=\textwidth]{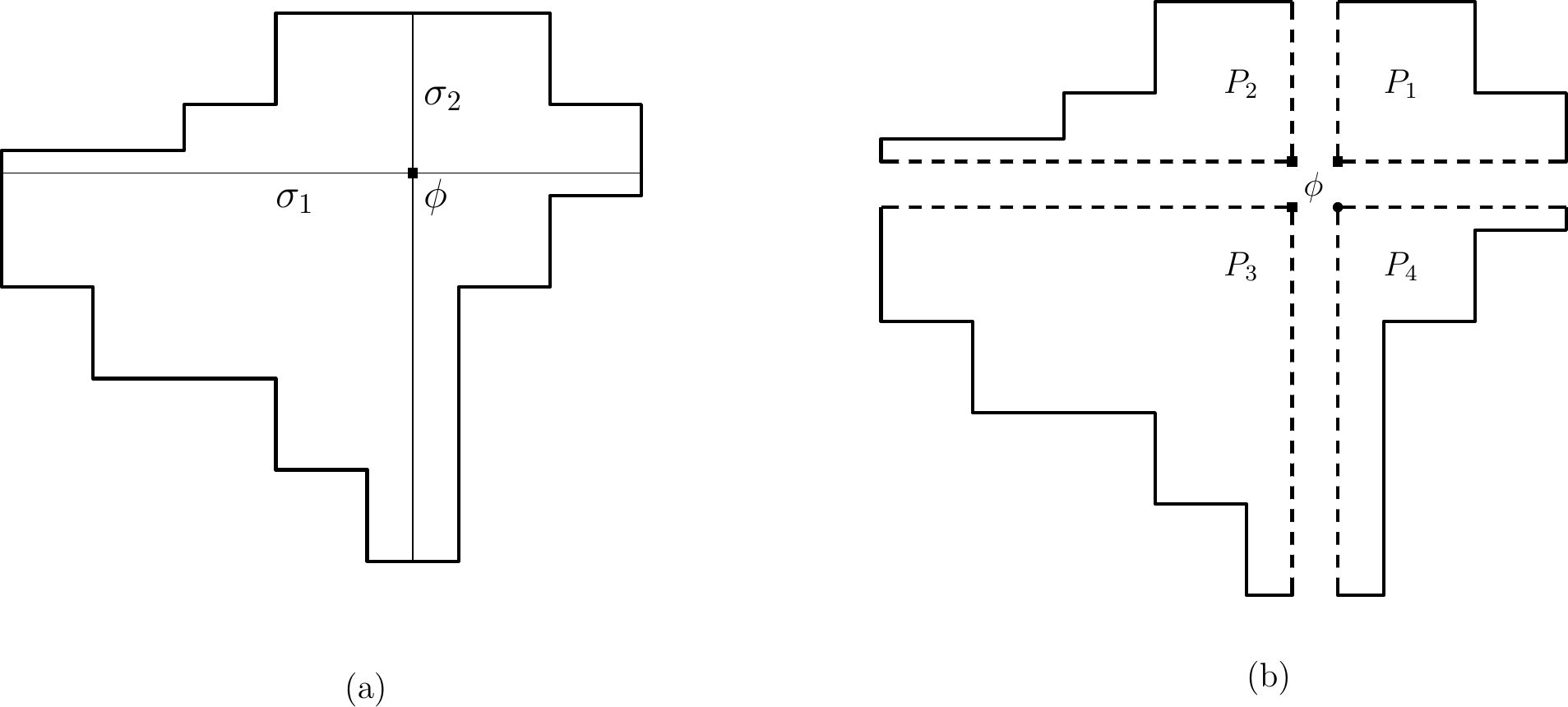}
	\caption{(a)An orthoconvex polygon $ P $ that has $ \sigma_{1} $ and $ \sigma_{2} $. Point $ \phi $ is the intersection of $ \sigma_{1} $ and $ \sigma_{2} $. (b)The decomposition of $ P $ into orthogonal fan polygons $ P_{1}$,$ P_{2}$,$ P_{3}$ and $ P_{4}$.}
	\label{fi:fig2}
\end{figure}
\begin{lemma}
 \label{le:lemma001}
   For any orthoconvex polygon $ P $, if there exists a horizontal line segment $ \sigma_{1} $  which is connecting the leftmost and the rightmost vertical edges of $ P $ such that $ \sigma_{1}\in P $ and there exists a vertical line segment $ \sigma_{2} $  which is connecting the upper and the lower horizontal edges of $ P $ such that $ \sigma_{2}\in P $, then $ P $ has a kernel. If guard $ g $ occurs in the kernel, every point in $ P $ is guarded by it. See figure~\ref{fi:fig2}.
\end{lemma}
\begin{proof}
If $ \sigma_{1} $ connects the leftmost and the rightmost vertical edges of $ P $ and $ \sigma_{2} $  connects the upper and the lower horizontal edges of $ P $ then they have an intersection $\phi$ that is contained in $ P $. $\sigma_{1} $ and $ \sigma_{2} $ decompose $P$ into 4 sub-polygons $ P_{1}$,$ P_{2}$,$ P_{3}$ and $ P_{4}$. All of obtained sub-polygons are orthogonal fan polygons and $\phi$ is their core vertex, jointly. In every part, the entire sub-polygon is visible from  $ \phi $ and also it is on the kernel. Therefore, $ \phi $ belongs to the kernel of $ P $ and if guard $ g $ occurs in the kernel, every point in $ P $ is guarded by it.
\end{proof}
Every polygon $ P $ divides the plane into 3 regions, an interior region bounded by the polygon, an exterior region containing all of the nearby and far away exterior points and boundary region containing all points on the boundary of polygon as denoted $int(P)$, $ext(P)$ and $ bound(P) $ so that $ P=int(P) \cup bound(P) $. If $ e $ be a line segment, $ int(e) $ is $ e $ without its two endpoints. If $ e $ be horizontal, then $ left(e) $ and $ right(e) $ are referred to left and right endpoints of $ e $, respectively. Also, if $ e $ be vertical, then $ top(e) $ and $ down(e) $ are referred to top and down endpoints of $ e $, respectively. The \textit{bounding box} of a polygon is the axis-aligned minimum area rectangle(box) within which all the points of polygon lie. If an $ x $-monotone orthogonal polygon and its bounding box has a horizontal edge in common, completely, the polygon is called \textit{histogram}, this common edge is called \textit{base}. If an orthoconvex polygon and its bounding box has a horizontal edge in common, completely, the polygon is called \textit{pyramid}, this common edge is called \textit{base}, too.

\section{An Algorithm for Guarding Monotone Art Galleries}
In the following, we present a linear-time exact algorithm for guarding orthogonal and $ x $-monotone polygons. Our algorithm uses a geometric approach instead of graph theoretical approach to obtain the result. Therefore, We can find the exact geometric locations of the point guards. 
\subsection{The Decomposition of an $ x $-monotone Orthogonal Polygon into the Balanced Ones}
Let $ P $ be an orthogonal $ x $-monotone polygon with $ n $ vertices. Start from the leftmost vertical edge of $ P $, as denoted, $ \varepsilon $ , propagate a light beam in rectilinear (straight-line) path perpendicular to $ \varepsilon $ and therefore collinear with the $ X $ axis. All or part of this light beam passes through some rectangles of set $ R $(name this subset $ R_{\rho} $) and these rectangles together make a sub-polygon $ \rho $ of $ P $. For every $ u_{i}$ and $l_{i} $ belongs to $ \rho $, it is established that $ \min_{u_{i} \in \rho} (y(u_{i}))\geq \max_{l_{j} \in \rho}(y(l_{j})) $. So, there exists a horizontal line segment $ \sigma $  which is connecting the leftmost and rightmost vertical edges of $ \rho $ such that $ \sigma$ contains in $\rho $, completely. If a polygon like $ \rho $ has this property, we say the polygon is \textit{balanced}. If we want to guard a balanced $ x $-monotone polygon with minimum number of guards, it is possible to set all guards on its $ \sigma $. So, If we remove $ \rho $ from $ P $ and iterate the described operations, $ P $ is decomposed into several balanced $ x $-monotone polygon. Now, we describe a linear-time algorithm for decomposition $ P $ to the balanced sub-polygons.
\begin{figure}
	\centering
	\includegraphics[width=\textwidth]{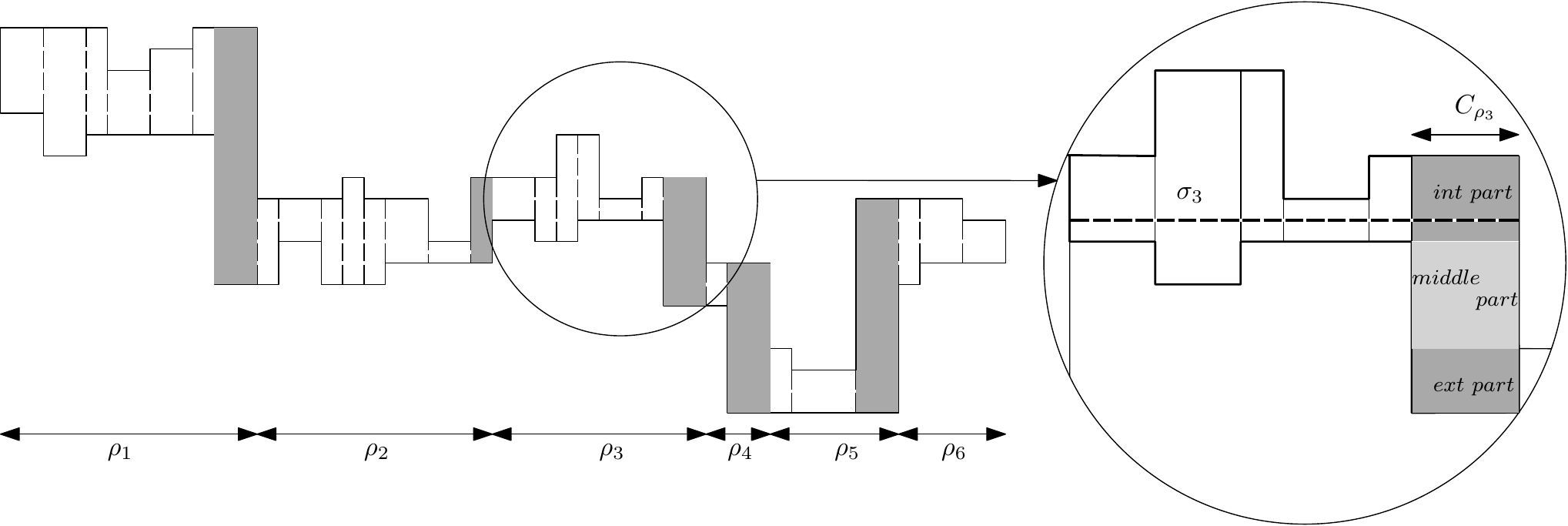}
	\caption{(a)Decomposition of monotone polygon into balanced sub-polygons and their vertical decomposition. Dark gray rectangles are cut. (b) A balanced sub-polygon and its align segment, cut rectangle part of sub-polygons has 3 disjoint parts.}
	\label{fi:fig3}
\end{figure}
\begin{algorithm}[]
	\KwData{an $ x $-monotone polygon with $ n $ vertices}
	\KwResult{the minimum number of balaced $ x $-monotone polygon}
	(1)Set $ min_u=u_1$ and $ max_l=l_1 $\;
	(2)\ForEach{rectangle $ R_{i} $ belongs to $ R $}{
		(3)\If{ $ u_i>max_l $ or $ l_i<min_u $}{
			remove $ R_1,\dots,R_{i-1} $ from $ R $\;
			refresh the index of  $ R $ starting with $ 1 $\;
			go to 1\;
		}
		
		(4)Compute $min_u=\min(min_u,u_{i})$ and $max_l=\max(max_l,l_{i})$\;
}
\caption{Decomposition $ P $ into the balanced sub-polygons.}
\label{al:algo1}
\end{algorithm}
If the condition in item 3 is satisfied, a balanced sub-polygon $ \rho $ is determined. So, we remove it from $ P $ and iterate algorithm for $ P-\rho $. We remove the rectangles belong to $ \rho $ from $ R $. We know the members of $ R $ are ordered from left to right and labeled from $ 1 $, After removing, we relabel the remained members from $ 1 $, again, to simplify the description of the algorithm. Certainly, the same processes will be occurred for $ U $ and $ L $. The number of iterations is equal to the cardinality of $ R $ (in the start). Therefore, the time complexity for the decomposition $ P $ into balanced polygons is linear. Remember the vertical decomposition of the polygon into rectangles, the last rectangle of each balanced sub-polygons (excepted the last sub-polygon) is named \textit{cut rectangle}. The cut rectangle of the sub-polygon $ \rho $ is denoted as $ C_{\rho} $. For an illustration see figure~\ref{fi:fig3}(a). Every cut rectangle has 3 parts which are obtained by extending horizontal edge of two adjacent rectangles, the one has intersection with the align segment is named int-part, two other ones are named middle-part and ext-part, see figure~\ref{fi:fig3}(b). The decomposition $ P $ into the minimum number of balanced sub-polygons is not unique. Assume that $ P $ is decomposed into the balanced sub-polygons $ \rho_{1},\rho_{2},\dots,\rho_{k} $, for every $ \rho_i $ and $ \rho_{i+1} $ we can replace them with $ \rho_i-C_{\rho_i} $ and $ \rho_{i+1}\cup C_{\rho_i}$. Bt this change, both of them remained balanced. But, we do this change only in certain circumstances that leads to simplicity. The last rectangle of $ \rho_i $ is $ C_{\rho_i}=R_{x+1}\in R $ and its previous rectangle be $ R_x \in R $. In every iteration of algorithm~\ref{al:algo1}, if $ R_x $ is not source and $ R_x $ is local minimum then do $ \rho_i=\rho_i-C_{\rho_i} $ and $ \rho_{i+1}=\rho_{i+1}\cup C_{\rho_i}$. Later, we will realize that the guarding of these two is more cost effective. Therefore, we modify algorithm~\ref{al:algo1} to algorithm~\ref{al:algo2}.
\begin{algorithm}[]
	\KwData{an $ x $-monotone polygon with $ n $ vertices}
	\KwResult{the minimum number of balaced $ x $-monotone polygon}
	(1)Set $ min_u=u_1$ and $ max_l=l_1 $\;
	(2)\ForEach{rectangle $ R_{i} $ belongs to $ R $}{
		(3)\If{ $ u_i>max_l $ or $ l_i<min_u $}{
			\eIf{$ i-1\neq 1 $ and $ h_{i-1}>h_i $ and $ h_{i-1}>h_{i-2} $}{
			   remove $ R_1,\dots,R_{i-1} $ from $ R $\;}{
			   remove $ R_1,\dots,R_{i-2} $ from $ R $\;
			   }
			 
			refresh the index of  $ R $ starting with $ 1 $\;
			go to 1\;
		}
		
		(4)Compute $min_u=\min(min_u,u_{i})$ and $max_l=\max(max_l,l_{i})$\;
}
\caption{The modified algorithm for decomposition $ P $ into the balanced sub-polygons.}
\label{al:algo2}
\end{algorithm}  
Every balanced polygon like $ \rho $ has a horizontal line-segment like $ \sigma $ which is connecting the leftmost and rightmost edges of $ \rho $ that is called \textit{align segment}. So, the entire $ \rho $ is visible from at least one point of $ \sigma $, i.e., $ \rho $ is weak visible from $ \sigma $.  Assume that $ P $ is decomposed into the balanced sub-polygons $ \rho_{1},\rho_{2},\dots,\rho_{k} $ and $ \sigma_{1},\sigma_{2},\dots,\sigma_{k} $ be their align segments, respectively. If $ 1\leq i,j\leq k $ and $\abs{i-j}>1$, for every $ p\in \sigma_i $ and $ q\in \sigma_j $, $ p $ and $ q $ is not visible from each others. So, if $j=i+1$, the only visible points from one segment to another is the right endpoint of $ \sigma_i  $ and the left endpoint of $ \sigma_j  $ i.e. for every point $ p\in int(\sigma_{i})$, there is no point belongs to $int(\sigma_{j}) $ that is visible from $ p $. Due to this fact, if we optimally cover $ P $ so that all the guards are located on the align segments, guarding each of sub-polygons can be done independently i.e. the minimum number of guards for guarding the entire polygon is the sum of the minimum number of guards that are necessary for every sub-polygons. 
\begin{claim}
There exists an optimal guard set $ G={g_1, g_2,\dots,g_{opt}} $ for a monotone orthogonal polygon $ P $ so that all guards are located on the align segments.
\end{claim}
To prove this claim, we present an algorithm in the next sections and prove that its output is optimal.

\subsection{The Algorithm for Guarding the Balanced Sub-polygons}
\label{ss:ss01}
Given a balanced $ x $-monotone orthogonal $ P $ with $ n $ vertices, we present algorithm~\ref{al:algo3} to guard $ P $ using the minimum number of guards. After vertical decomposition, it is obtained the sets $ R $, $ U $, $ L $, $ E_L $ and $ E_U $ for the polygon $ P $. Because of being balanced, $ P $ has an align segment $ \sigma $ which is connecting the leftmost and rightmost edges of $ P $ and contained in it.
\begin{defini}
For a horizontal edge $ e $ of the polygon $ P $, the set of every point $p\in P$ which there is a point $ q \in e $ such that $ pq $ is a line segment normal to $ e $ and completely inside $ P $, is named \textit{orthogonal shadow} of $ e $, as denoted $ os_e $(for abbreviation).
\end{defini}
First, we find all tooth edges of the set $ E = E_L \cup E_U $ and call the obtained set as $ D $. For every $  d_i \in D $, we compute orthogonal shadow of $ d_i $ as $ os_{i} $. Let $ D=\{d_1, d_2,\dots,d_k\} $ and $ OS=\{os_1,os_2,\dots,os_k\} $, ordered from left to right by $ x $-coordination of their left vertical edges.
\begin{lemma}
\label{le:le02}
Every tooth edge $ e_d $ can be guarded only with a guard that is located in its orthogonal shadow $ os_{e_d} $, not anywhere else. 
\end{lemma}
\begin{proof}
Proof by contradiction. Suppose that $ P $ is $ x $-monotone and the horizontal tooth edge $ e_t $ is guarded with $ g_t $ that is not located in $ os_{e_t} $, so, the point $ g_t $ is not in the $ x $-coordinate of any points on $ e_t $. Let the left and right endpoints of $ e_t $ be $ L_t $ and $ R_t $, respectively, and let $ x $-coordinate of $ g_t $ be less than $ x $-coordinate of $ L_t $ i.e. $ x(g_t)<x(L_t)$. The edge $ e_t $ is visible from $ g_t $, so, $ L_t $ and $ R_t $ are visible from $ g_t $. There is an axis-parallel rectangle spanned by the $ R_t $ and $ g_t $ is contained in $ P $. These two points do not have the same $ x $-coordinates and also $ y $-coordinates, $ 4 $ vertices of the rectangle are contained in $ P $ as $ R_t=(x(R_t),y(R_t)) $, $ A=(x(R_t),y(g_t)) $, $ B=(x(g_t),y(R_t)) $ and $ g_d=(x(g_t),y(g_t))$. So, the edge $ BR_t $ is contained in $ P $, too. It is impossible, because $ e_t \subset BR_t $ i.e. if an edge be a part of a line segment which is contained in the polygon, actually, it is not an edge.
\end{proof}
\begin{figure}
	\centering
	\includegraphics[width=\textwidth]{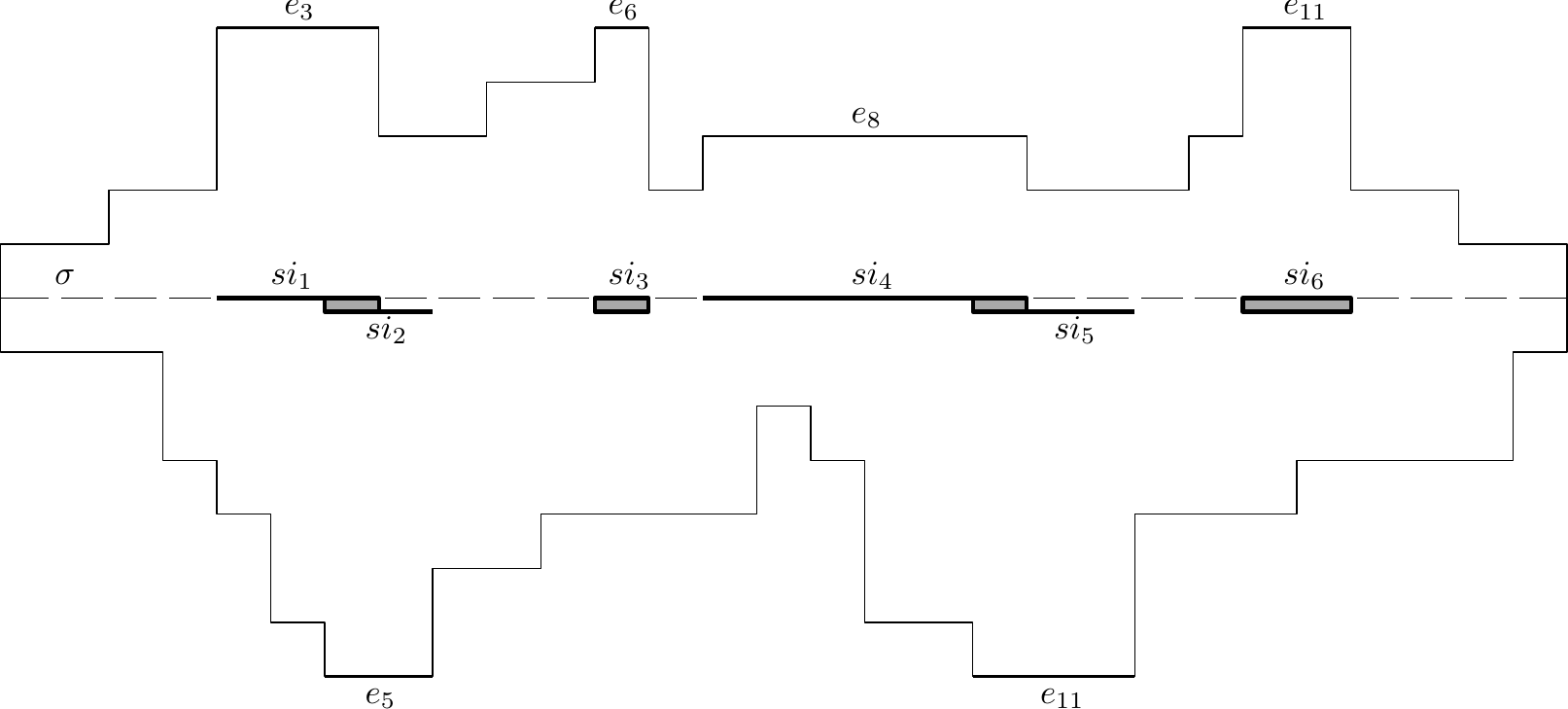}
	\caption{An Illustration of the algorithm, the bold edges are tooth and all bold parts of $ \sigma $ belong to $ SI $ that are the positions for guards.}
	\label{fi:fig4}
\end{figure}
If we want to guard the entire polygon $ P $, we must place guards so that there be at least one guard in every element of $ OS $. If tooth edge $ e_1 $ belongs to $ E_U $, another tooth edge $ e_2 $ belongs to $ E_L $ and $ os_{e_1}\cap os_{e_2}\neq \emptyset $ a guard on the intersection of them is sufficient to cover both of them. So, in the set $ OS $, if intersection of two members $ os_{e_1}$ and $ os_{e_2}$ be not empty then we remove them from $ OS $ and insert a new member as $ os_{e_1}\cap os_{e_2}  $. Note that before the replacements, the intersection of every $ 3 $ members of $ OS $ is empty and after them the cardinality of $ OS $ decrease to $ \kappa \leq k $. We want to place all guards on the align segment $ \sigma $ and the orthogonal shadow of every tooth edge has intersection with $ \sigma $. Let line segment $ si_i=os_i \cap \sigma $ and set $ SI=\{si_1,si_2,\dots,si_\kappa\} $ s.t. $ (\kappa \leq k) $, ordered from left to right by $ x $-coordination of their left endpoint. The intersection of every $ 2 $ members of $ SI $ is empty. The set $ SI $ is the positions for placing guards, one guard is located on an arbitrary point of every member of $ SI $, see figure~\ref{fi:fig4} Using this structure enable us to find the positions for locating the minimum number of guards in the balanced sub-polygon $ P $ in the linear time relative to $ n $. In algorithm~\ref{al:algo3}, the set $ SI $ is the positions for the optimum guard set and the variable $ m $ is the cardinality of the optimum guard set.
\begin{algorithm}[]
	\KwData{the horizontal edges of two chains of balanced $ x $-monotone polygon with $ n $ vertices ($ E_L $,$ E_U $)}
	\KwResult{the minimum number of guards ($ m $) and their positions ($ SI $)}
	Set $ m=0 $ and $ SI=\emptyset $\;
	\ForEach{horizontal edge $ e_{i} $ belongs to $ E_L $}{
		\If{ Angles of $ left(e_i)$ and $right(e_i)$ are equal to $\frac{\pi}{2}$}{
			$ A_i=(x(left(e_i)),y(\sigma) $\;
			$ B_i=(x(right(e_i)),y(\sigma) $\;
			Set segment $ si_i=A_iB_i $ and $ SI_L=SI_L \cup \{si_i $\}\;
			$ m++ $;
		}
	}
   \ForEach{horizontal edge $ e_{i} $ belongs to $ E_U $}{
		\If{ Angles of $ left(e_i)$ and $right(e_i)$ are equal to $\frac{\pi}{2}$}{
			$ A_i=(x(left(e_i)),y(\sigma) $\;
			$ B_i=(x(right(e_i)),y(\sigma) $\;
			Set segment $ si_i=A_iB_i $ and $ SI_U=SI_U \cup \{si_i\} $\;
			$ m++ $;
		}
   }
   Merge the sorted lists $ SI_L$ and $ SI_U$ as sorted list $ SI $\;
   \ForEach{horizontal segment $ si_{i} $ belongs to $ SI $}{
   		\If{$ si_i\cap si_{i+1}\neq \emptyset $}{
   			$ si_i=si_i\cap si_{i+1} $\;
   			$ SI=SI-\{si_{i+1}\}$\;
   			$ m-- $;
   		}
      }
\caption{Optimum guarding of a balanced polygon $ P $.}
\label{al:algo3}
\end{algorithm}
\\The positions of all guards are in the set $ SI $ and every elements of $ SI $ is a subset of align segment $ \sigma $, so, all guards are located on align segment $ \sigma $. The time complexity of the algorithm clearly is linear time corresponding to the cardinality of set $ E=E_L\cup E_U  $.
\begin{lemma}
\label{le:lemma3}
The minimum number of guards is equal to $ m $ that is obtained by algorithm~\ref{al:algo3} for guarding a balanced monotone orthogonal polygon $ P $.
\end{lemma}
\begin{proof}
Suppose that $ m $ guards is sufficient to cover the entire polygon $ P $, using lemma~\ref{le:le02} prove that this number of guards necessary even for guarding the dent edges of $ P $. Every line segment $ si_i\in SI $ is a subset of a r-star sub-polygon i.e. if we decompose $ P $ into r-star parts(sub-polygons) then the kernels of every r-star sub-polygons has at least one point in the elements of $ SI $, so the entire $ P $ is covered by these $ m $ guards and their positions.
\end{proof}

\subsection{Time Complexity of Algorithm}
In this subsection, we analyze the algorithm and describe how it can be implemented in linear time. To compute the optimal solution for guarding $ P $, we need to solve subproblems. To solve the subproblems of $ P $, we need to decompose $ P $ vertically into the set of rectangles $ R $ described in section~\ref{ss:ss1}. Therefore, we obtain the sets $ E $, $ U $ and $ L $. We show that this decomposition is possible in $ O(n) $-time. After that, we use algorithm~\ref{al:algo2} that is processing-able in $ O(n) $-time, we explain it before. The total of vertices of the all obtained balanced orthogonal is $ O(n) $, so, in algorithm~\ref{al:algo3}, finding the minimum number of guards for all balanced sub-polygons is possible in $ O(n) $-time. Therefore, all computations handle in $ O(n) $-time. Finally, $ m $ is returned as the optimal solution for the problem on $ P $. Therefore, we have proved the main result of this section:
\begin{theorem}
	\label{th:th01}
	There exists a purely geometric algorithm that can find the minimum number of guards for an orthogonal and $ x $-monotone polygon with $ n $ vertices, with orthogonal visibility in $ O(n) $ time.
\end{theorem}

\section{Hidden Guarding of Histogram Galleries}
Let $ P $ be a histogram polygon with $ n $ vertices, we want to find minimum number of guards such that every point in $ P $ is visible from some guards under the constraint that no two guards may see each other. Every histogram polygon is a balanced monotone one and has an edge that connecting the leftmost and rightmost vertical edges of it as called \textit{base}. So, the base edge is an align segment for histogram polygon $ P $. In every monotone polygon, the number of tooth edges is two more than the number of dent edges i.e the number of tooth edges belong to $ E_L $ (or $ E_U $) is one more than the number of dent edges belong to it. In this section, we present a linear-time exact algorithm for hidden guarding histogram polygons. Our algorithm uses a geometric approach instead of graph theoretical approach to obtain the result. Therefore, We find the exact geometric locations of the point guards. 

\subsection{The Decomposition of an histogram Polygon into Pyramid Polygons}
Given a histogram polygon $ P $ with $ n $ vertices. Extend every dent edge of $ P $, exclusively from its right endpoint until intersect the boundary. Using this strategy, decompose $ P $ into several sub-polygons. All the obtained sub-polygons are orthoconvex and absolutely pyramid polygons. If the number of dent edges of $ P $ be $ m $, then the number of obtained pyramid sub-polygons is $ m+1 $, exactly. The base edges of all pyramid polygons lie on the extended dent edges except the rightmost pyramid that its base edge is in common with the base of $ P $. For an illustration see figure~\ref{fi:fig5}(a).
\begin{figure}
	\centering
	\includegraphics[width=\textwidth]{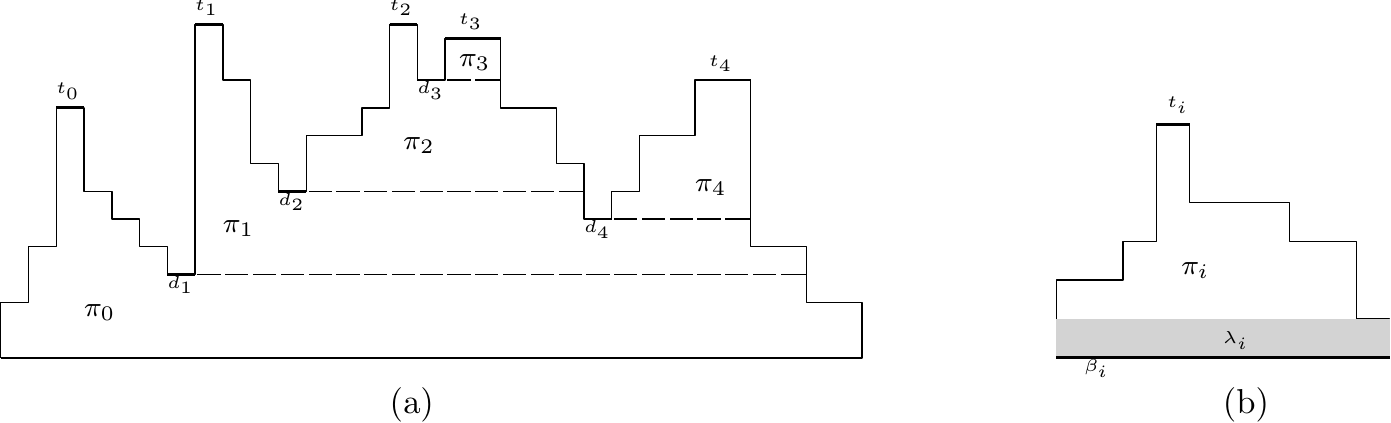}
	\caption{(a)The decomposition of histogram polygon into pyramid polygons. (b) The basis rectangle of a pyramid polygon is shown as gray.}
	\label{fi:fig5}
\end{figure}
\begin{defini}
In the pyramid polygon $ P_0 $, the maximum area rectangle $ R_0\subset P_0 $ that one of its edge is the base edge of pyramid, is named \textit{basis rectangle}. See figure~\ref{fi:fig5}(b).
\end{defini}
Suppose that $ \Pi=\{\pi_0, \pi_1,\dots, \pi_m\} $ is the set of all obtained pyramid sub-polygons ordered from left to right, according to the $ x $-coordinate of their leftmost vertical edges. Also, suppose that $ \beta_i $ and $ \lambda_i $ are the base edge and basis rectangle of sub-polygon $ \pi_i $, {\footnotesize for $ 0\leq i \leq m $}, respectively. In the pyramid $ \pi_i $, the base edge $ \beta_i $ is a tooth edge, except $ \beta_i $, there exists another tooth edge $ t_i $ on the upper chain of $ \pi_i $. For every pyramid sub-polygon, we compute the orthogonal shadow of $ t_i $ and intersection between $ \lambda_i $ and $ os_{t_i} $ as the position for placing a guard on it i.e. we hidden guarding an histogram polygon with locating one guard in one of the interior points of every $ si_i=\lambda_i \cap os_{t_i} $ as named \textit{shadow intersection} of $ \pi_i $.
\begin{lemma}
The interior points of two different shadow intersection are invisible from each other. In the other words, every two points $ p_1 \in int(si_i) $ and $ p_2\in int(si_j) $ that $ i\neq j $, is not visible from each other.
\end{lemma}
\begin{proof}
After the described decomposition, assume that $ \Pi=\{\pi_0, \pi_1,\dots, \pi_m\} $ be the set of obtained sub-polygons ordered from left to right, according to the $ x $-coordinate of their leftmost vertical edge. Each of them has a base edge and assume that $ B=\{\beta_0, \beta_1,\dots, \beta_m\} $ be the set of base edges of $ \Pi $ ordered corresponding to order of polygons they belong to. Also, let $ \Lambda=\{\lambda_0, \lambda_1,\dots, \lambda_m\} $  and $ SI=\{si_0, si_2,\dots, si_m\}$ be the sets of basis rectangles and shadow intersection areas of $ \Pi $ in the same expressed order. And assume that $ D=\{d_1, d_2,\dots,d_m\} $ be the set of dent edges of $ P $ ordered from left to right according to $ x $-coordinates of their left endpoints. Remember that $ \beta_x $ is obtained after extending the dent edge $ d_x $, therefore all the interior points of $ si_x $, $ \lambda_x $ and even $ \pi_x $ is higher than $ d_x $, {\footnotesize for every $ 0\leq x\leq m $}. Proof by contradiction. Suppose that there exist two points $ p_1 \in int(si_i) $ and $ p_2\in int(si_j) $ that $ i<j $, is visible from each other. So, there is an axis-aligned rectangle $ R $ spanned by $ p_1 $ and $ p_2 $ contained in $ int(P) $, then $ R\cap(P-int(P)) = \emptyset $, so, $ R\cap d_j=\emptyset $, too. Since that $ d_j $ belongs to upper chain, $ d_j $ is higher than all points in $ R $ and it is higher than $ p_2 $. So, there \textbf{is} a point belongs $ si_j $ that is higher than $ d_j $ and this is a contradiction.
\end{proof}
\subsection{Analyze of Algorithm}
In this subsection, we present pseudo-code of the algorithm that is described in the previous subsection as algorithm~\ref{al:algo4}. We analyze the algorithm, prove that it find the optimum number of hidden guards and explain how it can be implemented in linear time.
\begin{algorithm}[]
	\KwData{the horizontal and vertical edges of upper chain of histogram polygon with $ n $ vertices and base edge $ b $ (the set of horizontal edges $ E_H $, the set vertical edges $ E_V $)}
	\KwResult{the minimum number of hidden guards ($ m $) and their positions ($ SI $)}
	Set $ m=0 $ and $ SI=\emptyset $\;
	Set $ \varepsilon = $ minimum length of $ E_V $\;
	Set $ y_1=y(b) $ and $ y_2=y(b)+\varepsilon $\;
	\ForEach{horizontal edge $ e_{i} $ belongs to $ E_H $}{
		\If{ Angles of $ left(e_i)$ and $right(e_i)$ are equal to $\frac{3\pi}{2}$}{
			$ y_1=y(e_i) $\;
			$ y_2=y(e_i)+\varepsilon $\;
		}	
		\If{ Angles of $ left(e_i)$ and $right(e_i)$ are equal to $\frac{\pi}{2}$}{
			$ A_i=(x(left(e_i)),y_1) $\;
			$ B_i=(x(right(e_i)),y_2) $\;
			Set $ si_i= $ the rectangle spanned by two points $ A_i $ and $ B_i $\;
			Set $SI=SI \cup \{si_i $\}\;
			$ m++ $;
		}
	}
\caption{Optimum hidden guarding of a histogram polygon $ P $.}
\label{al:algo4}
\end{algorithm}
The interior points of the members of $ SI $ are invisible from each other, so, if we put a guard over each rectangle, then all the pyramid sub-polygons are guarded and also the entire histogram polygon is covered.
\begin{lemma}
The minimum number of hidden guards is equal to $ m $ that is obtained by algorithm~\ref{al:algo4} for guarding a histogram polygon $ P $.
\end{lemma}
\begin{proof}
The number of guards that is obtained using algorithm~\ref{al:algo4} for finding the minimum hidden guard set is equal to the number of guards that is obtained using algorithm~\ref{al:algo3} for the minimum regular guard set, therefore we proved that it is optimum in lemma~\ref{le:lemma3}. Every rectangle $ si_i\in SI $ is a subset of kernel of a pyramid sub-polygon that is r-star i.e. if we decompose $ P $ into r-star parts(sub-polygons) then the kernels of every r-star sub-polygons has at least one point in the elements of $ SI $, so the entire $ P $ is covered by these $ m $ guards and their positions.
\end{proof}
The time complexity of the algorithm is linear time corresponding to the size of the polygon, clearly. To solve the problem, we need to decompose $ P $ vertically into the set of rectangles $ R $ described in section~\ref{ss:ss1}. Therefore, we obtain the sets $ E_H $ and $ E_V $. This decomposition is possible in $ O(n) $-time. After that, we use algorithm~\ref{al:algo3} that is processing-able in $ O(n) $-time. So, all computations handle in $ O(n) $-time. Finally, $ m $ is returned as the optimal solution for the problem on $ P $. We have proved the main result of this section:
\begin{theorem}
There exists a purely geometric algorithm that can find the minimum number of hidden guards for a histogram polygon with $ n $ vertices, with orthogonal visibility in $ O(n) $ time.
\end{theorem} 

\section{Conclusion}
We studied the problem of finding the minimum number of hidden guards which is under orthogonal visibility. This new version is named \textit{hidden art gallery problem}. The total target in the hidden art gallery problem is finding the optimum hidden guard set $ G $ which is a set of point guards in polygon $ P $ that all points of the $ P $ are visible from at least one guard in $ G $ under the constraint that no two guards may see each other. We present an exact algorithm for finding the hidden guard set for histogram galleries. We solved this problem in the linear time according to $ n $ where $ n $ is the number of sides of histogram polygon. the space complexity of our algorithm is $ O(n) $, too. Many of the algorithms presented in this field are based on graph theory, but our proposed algorithm is based on geometric approach. This approach can lead to improved performance and efficiency in algorithms. For this reason, we also provided a purely geometric algorithm for the orthogonal art gallery problem (not hidden) where the galleries are monotone. We are aware that this problem has already been solved in linear time and our algorithm is linear-time, too. This new approach helped us solve the hidden art gallery problem more easily. Actually, the time complexity of the hidden orthogonal art gallery problem even for monotone polygon is still open. For the future works, we want to try to solve this problem for every simple orthogonal polygon with/without barriers. Both time and space complexity of our presented algorithm is order of $ O(n) $ and it is the best for this new version of the problem.

\bibliographystyle{plain}
\bibliography{bibfile}

\end{document}